\documentclass[12pt,a4paper]{article}
\usepackage{amsmath}
\usepackage{amsfonts}
\usepackage{amssymb}
\usepackage{amsthm}
\usepackage{graphicx}
\usepackage{tikz}
\usetikzlibrary{automata, positioning, arrows, calc,shapes}

\title{Completely Reachable Almost Group Automata}
\author{David Fernando Casas Torres}
\date{}

\newtheorem{proposition}{Proposition}

\newtheorem{lemma}{Lemma}
\newtheorem{corollary}{Corollary}
\newtheorem{definition}{Definition}

\newtheorem{theorem}{Theorem}
\newtheorem*{no-theorem}{Theorem}

\newcommand{\aut}[1]{\mathcal{#1}}
\newcommand{\nset}[1]{\{#1\}}

\newcommand{\kblock}[2]{{#1}^{[#2]}}
\newcommand{\coll}[1]{\mathrm{coll}(#1)}
\newcommand{\excl}[1]{\mathrm{excl}(#1)}
\newcommand{\dupl}[1]{\mathrm{dupl}(#1)}
\newcommand{\scc}{strongly connected component}
\newcommand{\compreach}{completely reachable}

\newcommand{\blim}{block of imprimitivity}
\newcommand{\blsim}{blocks of imprimitivity}
\newcommand{\gamgr}[2]{\Gamma_{#1}(#2)}
\newcommand{\leaf}[1]{\mathrm{leaf}(#1)}
\newcommand{\Sigmaast}{\Sigma^\ast}
\newcommand{\St}[2]{\mathrm{St}_{#1}(#2)}
\newcommand{\cor}[1]{\mathrm{Cr}(#1)}

\begin{document}
\maketitle
\begin{abstract}
We consider finite deterministic automata such that their alphabets consist of exactly one letter of defect 1 and a set of permutations of the state set.
We study under which conditions such an automaton is completely reachable. We focus our attention on the case when the set of permutations generates a transitive imprimitive group.  

%\keywords{Deterministic finite automata \and Transition monoid \and Complete reachability \and Permutation group}
\end{abstract}

\section{Introduction}

A deterministic finite automaton is said to be \emph{completely reachable} if every non-empty subset of states is the image of the whole state set by the action of some word.
 Such automata appeared in the study of descriptional complexity of formal languages~\cite{maslennikova2014resetcomplexity} and in relation to the \v{C}ern\'{y} conjecture~\cite{don2016vcerny}.
 A systematic study of completely reachable automata was initiated in~ \cite{bondar2016completely} and \cite{bondar2018characterization}, and continued in \cite{bondar2023completely}; in these papers Bondar and Volkov developed a characterization of completely reachable automata that relied on the construction of a series of digraphs.

One of the main results by Ferens and Szyku\l a  \cite{ferens2022completely} was an algorithm of polynomial time complexity, with respect to the number of states and letters, to decide whether a given automaton was completely reachable.
This seemed to solve the complexity problem for this kind of automata.
A different approach was proposed by Volkov and the author \cite{casas2022binary} for the special case of automata with two letters.
There the characterization relied on whether one of the letters acted as a complete cyclic permutation of the states and how the other letter acted on subsets of states representing subgroups of the cyclic group.

In this paper we give an approach to the generalization of the result in \cite{casas2022binary} by allowing that all the letters except one act as permutations of the set state and studying how the additional non-permutation letter acts on non-trivial blocks of imprimitivity if there are any.

The study of automata where all letters but one are permutations is by no means new.
This kind of automata is presented with different names.
In \cite{ananichev2019new}, they are called \textit{almost-permutation automata} and are used to present an example of a series of slowly synchronzing automata with a sink state.
In \cite{andre2004near}, automata are under the disguise of transformation semigroups and are called \textit{near permutation}.
In \cite{berlinkov2020synchronizing}, they are called \textit{almost-group automata}; there it is proved that these automata synchronize with high probability.
Finally, in \cite{rystsov2023primitive}, the non-permutation letters are the identity except in a subset of states where they have the same image.
There, it is proved that if no equivalence relation is preserved under the action of the letters, then the automaton is synchronizing.
Among these papers, we would like to highlight the work done in \cite{hoffman2023primitive} where the primitivity of a group of permutations of a state set has been tightly related to the complete reachability of the automata generated by adding a non permutation letter. Thus, the result presented here approaches this theory from the other side where the group is transitive but not primitive and we suggest a condition to ensure that the automaton generated is completely reachable.

In Section 2 we present the definitions and notation used in this paper. Then we proceed with the main results in Section 3.

\section{Preliminaries}

A \textit{deterministic finite automaton}, or simply an automaton, is usually defined as a triple $\aut{A} = \langle Q, \Sigma, \delta \rangle$, where $Q$, the states, and $\Sigma$, the alphabet, are finite sets and $\delta : Q \times \Sigma \rightarrow Q$ is the transition function.
For each letter in alphabet of the automaton $a \in \Sigma$ we can define the function $\delta_a : Q \rightarrow Q$ where $\delta_a(q) = \delta(q,a)$, hence each letter can be considered individually as a function of $Q$ on itself or a \emph{transformation} over $Q$.
This observation makes reasonable for us to use the following notation: for every $q \in Q$ and $a \in \Sigma$ we will denote $\delta(q,a) := q\cdot a $.
Derived from this we can say that an automaton can be specified just with its set of states and the action of each letter in this set; that is why from now on we will define automata as pairs of the state set and the alphabet.

A \emph{word}  of an automaton is a finite sequence of letters over its alphabet; this includes the empty word.
 The set of all words over the alphabet $\Sigma$ is denoted by $\Sigmaast$.
 We can extend the action of letters to words recursively in the following way: if $w \in \Sigmaast$, $a \in \Sigma$ and $q \in Q$, then $q \cdot wa := (q\cdot w) \cdot a$, and the action of the empty word is the identity function.
 Furthermore, the action of words can be applied to subsets of states: if $P \subseteq Q$ and $w \in \Sigmaast$, then $P \cdot w:= \nset{p \cdot w \mid \text{ for every } p \in P}$.

%Definition reachable and completelly reachable.
A subset of states $P \subseteq Q$ is called \emph{reachable} if there is a word $w \in \Sigmaast$ such that its image is exactly $P$, that is, $Q \cdot w = P$.
 An automaton is said to be \emph{synchronizing} if at least one singleton is reachable, i.e., there is a state $q \in Q$ and a word $w \in \Sigmaast$ such that $Q \cdot w = \nset{q}$.
 An automaton is \emph{completely reachable} if every non-empty subset of states is reachable.

%Definition excluded, defect and duplicated. 

Let $\aut{A} = \langle Q, \Sigma \rangle$ be an automaton and $w \in \Sigmaast$ an arbitrary word.
The excluded set of $w$ denoted by $\excl{w}$ is the set of states that have no preimages by $w$.
 The \emph{defect} of $w$ is the size of its excluded set.
 In the case the defect is 0, the word  $w$ represents a permutation of the set of states.
 Since a word is a total function, if the defect of $w$ is bigger than 0, then there must be states with the same image.
 These images are the \emph{duplicated} states of the word; the set of these states will be denoted by $\dupl{w}$. 
When any of these sets, $\excl{}$ or $\dupl{}$, is a singleton we will make no distinction between the set and the state inside it.
 Additionally, for the case of words of defect 1 we know that exactly two states must have the same image; we will call this pair of states the \emph{collapsed set} of the word, denoted by $\coll{}$.
 %Definition block of imprimitivity, imprimitive group.

Some transformations over a set of states $Q$ can be bijective, thus permutations.
Thanks to this we can use some terminology of the theory of permutation groups.
Recall that the set of all the bijective transformations of a finite set $Q$ on itself is denoted by $S_Q$, also called the \emph{symmetric group}.
Let $G \subseteq S_Q$ be a group of permutations of $Q$.
This group is said to be \emph{transitive} if  for every pair of states $q,p \in Q$ there is a permutation $\sigma \in G$ such that $p\cdot \sigma = q$.
A non-empty subset $B \subseteq Q$ is said to be a \emph{block} of the group if and only if for every $\sigma \in G$ either $B \cdot \sigma = B$ or $B \cdot \sigma \cap B = \emptyset$.
The singletons and $Q$ itself are, always, blocks, these are called \emph{trivial}.
A permutation group $G \subseteq S_Q$ is said to be \emph{primitive} if it is transitive and the only blocks are the trivial ones; otherwise the group is said to be \emph{imprimitive}.
In this article when we talk about a \blim, unless stated the contrary, it will always be non-trivial.
If a transitive group $G \subset S_Q$ has a \blim\ $B\subseteq Q$, all the images of $B$ by elements of $G$ are also blocks of imprimitivity and form a partition of $Q$.
This partition of subsets is called a \emph{system of imprimitivity} of the group $G$ over $Q$. 

%Grafos dirigidos, caminos y componentes conexas.
A \emph{directed graph} $\Gamma$ is a pair $(V(\Gamma),E(\Gamma))$, where $V(\Gamma)$ is the vertex set and $E(\Gamma) \subseteq V(\Gamma) \times V(\Gamma)$ is the set of directed edges.
To each edge we can assign one or more labels from some set $L$.
%Thus, we can have digraphs with parallel edges.
We will denote an edge $(s,t) \in E(\Gamma)$ labelled with $w$ as $s \xrightarrow{w} t$.
Since in this paper we consider only directed graphs, from now we omit the word ``directed''.
For reference, the first and second components of an edge will be called the \emph{source} and \emph{tail} respectively.
A \emph{path} of a graph is a set of edges $e_1, e_2, \dots, e_m$, $m\geq 1$, such that for every $1 \leq i < m$, the tail of $e_i$ is the same as the source of $e_{i+1}$. 
Vertices $p,q \in V(\Gamma)$ are \emph{ strongly connected} if there is a path from $p$ to $q$ and from $q$ to $p$.
Also we consider each vertex as strongly connected with itself.
A \scc\ of a graph is a maximal subset of vertices such that all its vertices are strongly connected to each other.

An automaton can be represented as a labelled graph, where the vertex set is the states set and for each state $p$ and letter $a$, there is a labelled edge $p \xrightarrow{a} p\cdot a$.
This is the \emph{underlying} graph of the automaton.

As we mentioned in the introduction, in \cite{hoffman2023primitive} the following characterization of primitive permutation groups is given.
Here $[n] := \nset{1, 2 \dots, n}$ and if $S$ is a set of transformations, then $\langle S \rangle$ is the transformation semigroup generated by $S$.
Recall that a transformation $f$ is \emph{idempotent} if $f^2 = f$. 

\begin{theorem}[\cite{hoffman2023primitive}[Theorem 3.1]]
Let $G$ be a permutation group on $[n]$ with $n \geq 3$.
Then $G$ is primitive if and only if for each [idempotent] transformation $f: [n] \rightarrow [n]$ of \emph{defect 1} every non-empty subset $A \subseteq [n]$ is reachable in $\langle G \cup \nset{f} \rangle$.
\end{theorem}

This theorem presents a characterization of primitive groups.
In the language of automata it states that in the presence of a set of permutation letters that generates a primitive group, the addition of any transformation of defect 1 suffices to obtain a completely reachable automaton.
Here we study a related case.
We would like to know what happens when the group generated by the permutation letters is transitive but not primitive.
We will see how this situation is not that forgiving and it requires a more complex relation between the group generated by the permutations and the transformation of defect 1.
The results proved in this article are closely related to the ones presented  in \cite{casas2022binary} for automata with just one permutation letter and one with defect 1.
The work presented in this article and in \cite{casas2022binary} is based on the work made in \cite{bondar2018characterization}; there the main actor is a graph constructed in several steps.
We will briefly explain the construction in Section \ref{sec:Rystsov graphs}.

For the rest of this paper we will consider automata $\aut{A} = \langle Q, \Sigma \rangle$, where $\Sigma = \Sigma_0 \cup \nset{a}$ and:
\begin{itemize}
\item The set of letters $\Sigma_0 \subset S_Q$, are all permutations of $Q$.
\item The generated subgroup $G = \langle \Sigma_0 \rangle$ is transitive.
\item The letter $a$ has defect 1.
\end{itemize}
The excluded state of the only letter of defect 1 will be denoted by $e$, i.e., $\excl{a} = e $.
Unless specified otherwise, the group generated by all permutation letters is denoted as $G$.
We will call automata with these characteristics \emph{almost group} automata.

Let $r \in \coll{a}$ be one of the two states collapsed by $a$.
There is a permutation that sends $e$ to $r$; call it $\sigma \in G$.
The transformation $\sigma a$ has defect 1, $e = \excl{\sigma a}$, and $e \in \coll{\sigma a}$.
Consider the automaton $\overline{\aut{A}} = \langle Q, \Sigma_0 \cup\nset{\sigma a} \rangle$.
Note that $\aut{A}$ is completely reachable if and only if $\overline{\aut{A}}$ is completely reachable.
Therefore, there is no loss of generality when we add the condition that $e \in \coll{a}$ from the beginning.
When this happens we call the automaton \emph{standardized}.
This change will simplify the arguments we use for the rest of the article.

%In this work when a \blim\ is mentioned the trivial cases, i.e., the whole set or singletons, are not considered.

For any automaton a subset of states,  $P \subseteq Q$ is \emph{invariant} by a transformation $w \in \Sigmaast$, or $w$-\emph{invariant}, if $P\cdot w \subseteq P$. 
The condition to get complete reachability in the binary case is that no subset of states that \textit{represents} a subgroup of the cyclic group is invariant by letter of defect 1.
%In the case of almost group automata this representation is not so clear.
The cosets of any subgroup generate a partition of the group that contains this subgroup.
There is a parallel situation in the case of \blsim, they form a partition of the set of states.
There are subgroups for each of these \blsim\  that let them invariant\footnote{Considerations of this are treated ahead in the paper.}.
This is the main reason to consider systems of imprimitivity. 
These partitions of the set of states are the closest to represent subgroups of the group acting on the states.

\section{A necessary condition} 
First, we begin by proving that complete reachability implies that some blocks can not be $a$-invariant.

\begin{proposition}
Let $\aut{A}$ be a standardized almost group automaton.
If $\aut{A}$ is \compreach , then $G$ is transitive and if there is at least one \blim\, then no \blim\ that contains $e = \excl{a}$ is invariant by $a$.
\label{Comp-Reach implies a-invariance}
\end{proposition}

\begin{proof}
First, let us prove that the condition for the group generated by the set of permutations to be transitive is necessary. For every word $w \in \Sigmaast$ it is true that $e \in \excl{wa}$; furthermore,
 $$\vert \excl{w} \vert \leq \vert \excl{w a} \vert \leq \vert \excl{w} \vert + 1.$$
This is, the action of adding $a$ to a word either increases by one or keeps the defect of the resulting word.
Note that adding a permutation does not modify the defect of any word.
Hence, in order to reach the subsets $Q \setminus \nset{q}$ for every $q \in Q$, it is necessary that there exists a permutation $\sigma_q \in G$ such that $e \cdot \sigma_q = q$. Let $p,q \in Q$ be an arbitrary pair of states. By the previously said, if $\aut{A}$ is completely reachable, then there are two permutations $\sigma_p,\sigma_q \in G$ such that $e \cdot \sigma_p = p$ and $e \cdot \sigma_q = q$. Finally, note that $p \cdot \sigma_p^{-1}\sigma_q = q$. Thus, $G$ is transitive.

For a subset of states $S \subset Q$, we denote by $\overline{S}$ its complement, i.e., $Q \setminus S$.

The proof that no \blim\ is $a$-invariant will be by contradiction. Suppose that $B$ is a \blim\ that contains $e$ and is $a$-invariant. This block belongs to a system of imprimitivity. Let $w \in \Sigmaast$ be the shortest word that reaches the complement of a block of this system, say $\overline{{C}}$. If $w = w^\prime \, b$ with $b \in \Sigma_0$, then $Q \cdot w^\prime = \overline{{C}}\cdot b^{-1}$, the complement of a \blim. This contradicts the condition of $w$ being the shortest word. Hence, $w$ can not end in a permutation.

As a consequence, the word $w$ ends with the letter $a$, i.e., $w = w^\prime \, a$. Recall that $Q \cdot w^\prime a \subset Q\cdot a$ and $e \notin Q\cdot a$, thus $e \notin Q\cdot w$ and we conclude that $Q \cdot w = \overline{B}$.

Since $B$ is $a$-invariant, we can conclude that its complement is also $a$-invariant. And since every $q \in \overline{B}$ has a preimage by $a$ then this letter acts as a permutation of $\overline{B}$. Therefore $Q\cdot w^\prime = Q\cdot w^\prime a = \overline{B}$ what, again, contradicts the supposition of $w$ being the shortest. There is no other type of letter in which the word $w$ could finish, then we end with an absurd. This situation came from supposing that $B$ is $a$-invariant, thus we have our proposition.
\end{proof}

By the preceding proof we have:
\begin{corollary}
If there is a \blim\ that contains $e$ and is invariant by $a$, then its complement is not reachable.
\end{corollary}

\section{Rystsov graphs of almost group automata}
\label{sec:Rystsov graphs}
\subsection{The structure of Rystsov graphs}
In \cite{bondar2018characterization} and \cite{bondar2023completely} Bondar and Volkov presented a characterization of completely reachable automata.
The characterization relies on the construction of a graph that can be constructed from an arbitrary automaton.  
This graph is a generalization of the ideas presented by Igor Rystsov in \cite{rystsov2000idempotents}.
That is why we will call these graphs as \emph{Rystsov graphs} of the automata.

The Rystsov graph of an automaton $\aut{A}$, denoted by $\Gamma(\aut{A})$, is constructed in an inductive way.
This means that in order to construct $\Gamma(\aut{A})$ we first assemble a graph called $\Gamma_1(\aut{A})$, verify if we can continue and if that is the case from $\Gamma_1(\aut{A})$ we continue with the construction of $\Gamma_2(\aut{A})$ and so on.
This series of graphs is guaranteed to always finish, the final graph is the Rystsov graph of $\aut{A}$. 
For the construction of the partial graphs we make use of the sets of words $W_k(\aut{A}) \subset \Sigmaast$, defined as the subset of all words of defect $k$ for $k \geq 1$.

The first step is to construct the graph $\gamgr{1}{\aut{A}}$ where its vertex set is $Q_1 := Q$ and its edge set is defined as:
$$ E_1(\aut{A})  := \nset{\excl{w} \xrightarrow{w} \dupl{w} \in Q_1 \times Q_1 \mid w \in W_1(\aut{A})}.$$
%TODO Espacio para dos ejemplos uno donde gamma 3 sea conex y otro donde no lo sea?
\textsl{Example}

Consider the automaton $\aut{E}_{18} := \langle \nset{1, 2, \dots, 18}, \nset{a,b,c} \rangle$, where $b$ and $c$ are permutations with the following cyclic representation:
\begin{align*}
b :=& (1,11,13,5,7,17)(2,10,14,4,8,16)(3,12,15,6,9,18)\\
c:= & (1,3,2)(4,5,6)(7,13)(8,16)(9,15)(10,14)(11,17)(12,18),
\end{align*}
and the transformation $a$ has defect 1. The following representation of $a$ puts the respective image under each state and omits the states that do not change:
$$
\left( {1\; 2\; 5\; 6\;  8} \atop {6\; 8\;  6\; 5\; 2} \right).
$$
Note that the $\excl{a} = 1$, $\dupl{a} = 6$ and $\coll{a} = \nset{1,5}$. The group generated by $\nset{b,c}$ is transitive and the blocks of imprimitivity that contain the state $1$ are the sets
$$\nset{1,5}, \nset{1,2,3,4,5,6}.$$
We can see $1 \xrightarrow{a} 6$, $2 \xrightarrow{ac^2} 5$ and $1 \xrightarrow{ab^3a} 3$ are edges in $E_1(\aut{E}_{18})$.\\

We continue the definitions. For any automaton $\aut{A}$ consider the following subset of states:
$$
D_1(\aut{A}) := \nset{p \in Q_1 \mid p = \dupl{w}\quad \&\quad e = \excl{w} \text{ for } w \in W_1(\aut{A})}.
$$

These are the states directly connected to $e$ in $\gamgr{1}{\aut{A}}$, that is, tails of the edges with $e$ as source. The following lemma states that all the edges in $\gamgr{1}{\aut{A}}$ are images by $G$ of these initial edges.
 
\begin{lemma}
If $q \to p \in E(\gamgr{1}{\aut{A}})$, then there are $\sigma_q \in G$ and $d \in D_1(\aut{A})$ such that $e\cdot \sigma_q = q$ and $d \cdot \sigma_q = p$. Or, what is equivalent, there is a permutation $\sigma_q \in G$ such that $p\cdot \sigma_q^{-1} \in D_1(\aut{A})$.
\label{edges gamma 1}
\end{lemma}

\begin{proof}
If $q \to p$ is an edge of $\gamgr{1}{\aut{A}}$, then there is a word of defect 1, call it $w \in W_1(\aut{A})$ such that $\excl w = q$ and $\dupl w = p$; this happens due to the definition of $\gamgr{1}{\aut{A}}$.
Remember that $G$ is transitive, thus there is a permutation $\sigma_q \in G$ such that $e \cdot \sigma_q = q$.

The word $w\sigma_q^{-1}$ has defect 1 and $\excl{w \sigma_q^{-1}} = \excl{w} \cdot \sigma_q^{-1}$, at the same time $\dupl{w\sigma_q^{-1}} = \dupl{w}\cdot \sigma_q^{-1}$.
Thus $\excl{w\sigma_q^{-1}} = q \cdot \sigma^{-1} = e$ and $p \cdot \sigma_q^{-1} \in D_1(\aut{A})$.$\qed$
\end{proof}

This lemma also tell us that in order to compute $\gamgr{1}{\aut{A}}$ it is sufficient to calculate $D_1(\aut{A})$, and then apply to the generated edges permutations that send $e$ to each of the different states of the automaton. In our running example the initial edges of $\gamgr{1}{\aut{E}_{18}}$ are shown in Figure \ref{Initial edges}.

\begin{figure}[h]
\centering
\begin{tikzpicture}
	\node[circle, draw] at($(0,0)$) (1) {$1$};
	\node[circle, draw] at($(2, 1)$) (6) {$6$};
	\node[circle, draw] at($(2, 0)$) (5) {$5$};
	\node[circle, draw] at($(2,-1)$) (3) {$3$};
	
	\draw
	(1) edge[-latex, above] node[sloped]{$a$} (6)
	(1) edge[-latex, above] node[sloped]{$a^2$} (5)
	(1) edge[-latex, above] node[sloped]{$ab^3a$} (3)
	;
\end{tikzpicture}

\caption{The initial edges of $\gamgr{1}{\aut{E}_{18}}$.}
\label{Initial edges}
\end{figure}

And the \scc\ that contains $1$ is shown in Figure \ref{strongly connected} (we omitted the labels to avoid confusion).

\begin{figure}[h]
\begin{center}
\begin{tikzpicture}
	\foreach \i in {1,2,3,4,5,6}
	{
		\node[fill=white, circle, draw, scale=0.75] at ($({60*(\i -1) }:2cm) + (-3,0)$) (\i) {$\i$};
	}
	
	\draw
	(1) edge[-latex] (6)
	(1) edge[-latex] (5)
	(1) edge[-latex] (3)
	
	(2) edge[-latex] (5)
	(2) edge[-latex] (4)
	(2) edge[-latex] (1)
	
	(3) edge[-latex] (4)
	(3) edge[-latex] (6)
	(3) edge[-latex] (2)
	
	(4) edge[-latex] (3)
	(4) edge[-latex] (2)
	(4) edge[-latex] (5)
	
	(5) edge[-latex] (3)
	(5) edge[-latex] (1)
	(5) edge[-latex] (6)
	
	(6) edge[-latex] (1)
	(6) edge[-latex] (3)
	(6) edge[-latex] (4)
	;
\end{tikzpicture}

\caption{A \scc\ of $\gamgr{1}{\aut{E}_{18}}$.}
\label{strongly connected}
\end{center}

\end{figure}

Now, let $\kblock{C_e}{1} \subseteq Q_1$ be the vertex set of the \scc\ of $\gamgr{1}{\aut{A}}$ that contains $e$.

\begin{lemma}
The set $C_e^{[1]}$ is a \blim.
\label{strongly connected is block 1}
\end{lemma}

\begin{proof}
If $\gamgr{1}{\aut{A}}$ is strongly connected then $C_e^{[1]} = Q$ and the proposition is true.
Then let us assume $\Gamma_1(\aut{A})$ is not strongly connected and $C_e^{[1]}$ is a proper subset of $Q$.

Let $\sigma \in G$ be a permutation such that $e \cdot \sigma = \dupl{a} = d$, then the edge $d \xrightarrow{a \sigma} d\cdot \sigma \in E_1$. If we repeat the application of $\sigma$, the produced words have all defect 1, and there is an $i \geq 1$ such that $d \cdot \sigma^i = e$. Then, $C_e^{[1]}$ is not a singleton since at least $e, d \in C_e^{[1]}$.

Considering this  we will prove first that for any $\rho \in G$, the subset $\kblock{C_e}{1}\cdot \rho$ is also a \scc.

Let $p,q \in \kblock{C_e}{1}$ be two arbitrary states. In $\gamgr{1}{\aut{A}}$ there is a path:
$$
p \xrightarrow{w_1} t_1 \xrightarrow{w_2} t_2 \dots t_{k-1} \xrightarrow{w_k} q
$$
where every $w_i$ is a word of defect 1. Since permutations act well on excluded and duplicated states, then:
$$
p \cdot \rho \xrightarrow{w_1 \rho} t_1 \cdot \rho \xrightarrow{w_2 \rho} t_2 \cdot \rho \dots t_{k-1}\cdot \rho \xrightarrow{w_k \rho} q\cdot \rho
$$
is a path in $\kblock{C_e}{1} \cdot \rho$; in the same way we can prove the existence of a path connecting $q\cdot \rho$ with $p \cdot \rho$, making $\kblock{C_e}{1} \cdot \rho$  a \scc.

What is left is to prove that $\kblock{C_e}{1}$ and its images by permutations of $G$ are blocks of imprimitivity.
Let $\rho \in G$ be an arbitrary permutation.
Suppose $x \in \kblock{C_e}{1} \cap \kblock{C_e}{1} \cdot \rho$ and let $y \in \kblock{C_e}{1}$ and $z \in \kblock{C_e}{1}\cdot \rho$ be two different states.
By the definition of \scc\ there are paths from  $y$ to  $x$, from $x$ to $z$ and, going back, from $z$ to $x$, and from $x$ to $y$.
Then $$\kblock{C_e}{1} = \kblock{C_e}{1} \cdot \rho.$$
This makes $\kblock{C_e}{1}$ a \blim.%$\qed$
\end{proof}

We continue the inductive construction of the graph $\gamgr{}{\aut{A}}$. Once we get $\gamgr{k}{\aut{A}}$, $k \geq 1$, if one of the following alternatives happens then the construction will be stopped and $\gamgr{}{\aut{A}} := \gamgr{k}{\aut{A}}$: either the graph is strongly connected, or all the \scc s are not \textit{big enough} (we will address the meaning of this in a moment). If none of these two possibilities happen, then we  proceed to construct $\gamgr{k+1}{\aut{A}}$. The new vertex set $Q_{k+1}$ will consists of the \scc s of $\gamgr{k}{\aut{A}}$; thus, each vertex is a collection of vertices of the set $Q_k$.

 In order to define the edges of this new graph, we need to properly define when a \scc\ is \textit{big enough}. For this note that  each vertex of $\gamgr{2}{\aut{A}}$ is a subset of states (even considering singletons) and the vertices of $\gamgr{3}{\aut{A}}$ would be collections of subsets of states and so on. With this in mind for $k\geq 2$ let $V \in Q_k$ be a vertex of $\gamgr{k}{\aut{A}}$, define the \emph{foliage} of $V$, and denote it by $\leaf{V}$, as follows: for $V \in Q_2$,  its foliage is the set itself, i.e., $\leaf{V} := V$, and for $k > 2$,
 $$\leaf{V} := \bigcup_{x \in V} \leaf{x} .$$

At the end $\leaf{V}$ is a subset of states. A vertex $V$ of $\gamgr{k+1}{\aut{A}}$, or, what is the same, a \scc\ of $\gamgr{k}{\aut{A}}$, is \textit{big enough} if $\vert \leaf{V} \vert \geq k+1$.
Thus, we stop the construction if none of the would be vertices of $\Gamma_{k+1}{\aut{A}}$ have more of $k+1$ states in their foliages.
(The term ``foliage'' is borrowed from \cite{bondar2023completely}, where the definition of the vertex sets of the graphs $\Gamma$ takes form of a rooted tree.)
Suppose that this is not the case and we can continue the process, then we can define a new set of edges: 
\begin{multline*}
E_{k+1} := \nset{C \xrightarrow{w} D \in Q_{k+1} \times Q_{k+1} \mid  C \neq D,\; \text{there is a } w \in W_{k+1}(\aut{A}),\\ \excl{w} \subseteq \leaf{C},\; \dupl{w} \cap \leaf{D} \neq \emptyset}.
\end{multline*}

The edge set of $\gamgr{k+1}{\aut{A}}$ will be the edges of $\gamgr{k}{\aut{A}}$ that connect different vertices of $Q_{k+1}$ together with the set $E_{k+1}$. For a more detailed discussion of the construction of $\gamgr{}{\aut{A}}$ we recommend the reader \cite[Section 3]{bondar2023completely}.  

We have the following theorem that characterizes completely reachable automata:

\begin{theorem}{\cite{bondar2018characterization}}
If an automaton $\aut{A} = \langle Q, \Sigma, \rangle$ is such that the graph $\gamgr{}{\aut{A}}$ is strongly connected and $\gamgr{}{\aut{A}} = \gamgr{k}{\aut{A}}$, then $\aut{A}$ is completely reachable; more precisely, for every non-empty subset $P \subseteq Q$, there is a product $w$ of words of defect at most $k$ such that $P = Q\cdot w$.
\label{bondar2018}
\end{theorem}

In the case that the group $G$ is primitive over $Q$, from Lemma \ref{strongly connected is block 1} we can see that $\gamgr{1}{\aut{A}}$ will be strongly connected and by Theorem \ref{bondar2018} it immediately follows that $\aut{A}$ is \compreach . That is why, from now on the group $G$ will, besides being transitive, have at least a \blim .\\ \\
\textsl{Continuation of the example}

Recall the automaton $\aut{E}_{18}$. We have seen that $C_e^{[1]} = \nset{1,2,3,4,5,6}$, and the other \scc s are the sets $B_2 =\nset{7,8,9,10,11,12}$ and $B_3 =  \nset{13,14,15,16,17,18}$. Since these sets have more than two elements, we can continue the construction of $\gamgr{}{\aut{E}_{18}}$. Accordingly to the previously said, the vertex set of $\gamgr{2}{\aut{E}_{18}}$ is  $Q_2 = \nset{C_e^{[1]}, B_2, B_3}$.  Consider the word $w := ab^3aca$, note that $\excl{w} = \nset{1,3}$ and $\dupl{w} = \nset{8,6}$, hence the edge $C_e^{[1]} \xrightarrow{w} B_2 \in E_2$. If we add $b$ twice more we have:
\begin{align*}
\excl{wb} = \nset{11,12},\; & \dupl{wb} = \nset{9,16}\\
\excl{wbb} = \nset{13,15},\; & \dupl{wbb} = \nset{18,2}.
\end{align*}

Thus adding the edges $B_2 \xrightarrow{wb} B_3$ and $B_3 \xrightarrow{wbb} C_e^{[1]}$ to $E_2$. These are enough to conclude, thanks to Theorem \ref{bondar2018}, that $\aut{E}_{18}$ is completely reachable.\\

%$$
%BD_k(\aut{A}) := \nset{ X \in Q_k \mid D_k(\aut{A})\cap \leaf{X} \neq \emptyset}.
%$$
We will extend the results given by Lemma \ref{edges gamma 1} and Lemma \ref{strongly connected is block 1}. Following the previous notation denote the \scc\ that contains $e$ in the graph $\gamgr{k}{\aut{A}}$ as $C_e^{[k]}$.

\begin{lemma}
If the foliages of the vertices in $\gamgr{k}{\aut{A}}$ form a system of imprimitivity of $G$ over $Q$, then $Y \rightarrow Z \in E_{k+1}$ if and only if there is a permutation  $\sigma \in G$ and a set $X \in Q_{k}$ such that
$\leaf{Y} = \leaf{\kblock{C_e}{k}} \cdot \sigma$; $\leaf{X} \cdot \sigma = \leaf{Z}$ and $\kblock{C_e}{k} \rightarrow X \in E_{k+1}$. 
\label{edges gamma k}
\end{lemma}

\begin{proof}
Since permutations respect the defect of any word and act well on excluded and duplicated sets, the converse is easy to see.

Now, if $Y \xrightarrow{w} Z \in E_{k+1}$, with $w \in W_{k+1}(\aut{A})$, then  $\excl{w} \subset \leaf{Y}$ and $\dupl{w} \cap \leaf{Z} \neq \emptyset$.
Let $w = u\, a \, \sigma$ with  $\sigma \in G$, this is, $\sigma$ is the longest permutation after the last appearance of the letter $a$ in $w$.
Since permutations do not increase the defect of a word, then $ua \in W_{k+1}(\aut{A})$ and $\excl{w \sigma^{-1}} = \excl{ua}$.
From the last affirmation we can conclude that $\excl{ua} \subseteq \leaf{Y}\cdot \sigma^{-1}.$

Since, by hypothesis, $\leaf{Y}$ is a \blim\ then also it is $\leaf{Y}\cdot \sigma^{-1}$.
Recall that $e \in \excl{ua}$ thus $\excl{ua} \subseteq \kblock{C_e}{k} = \leaf{Y}\cdot \sigma^{-1}$.
Using the same argument we can conclude that $\leaf{X} = \leaf{Z} \cdot \sigma^{-1}$.%$\qed$
\end{proof}

\begin{lemma}
If the foliages of the vertices in  $\gamgr{k}{\aut{A}}$ form a system of imprimitivity, then the foliage of $\kblock{C_e}{k+1}$ is a \blim\ of $G$ over $Q$.
\label{strongly connected is block k}
\end{lemma}

\begin{proof}
If each of the foliages of the vertices of  $\gamgr{k}{\aut{A}}$ forms a system of imprimitivity, then the foliage of $\kblock{C_e}{k+1}$ is just the union of \blsim.\\
We can use an argument similar to the one used in the proof of Lemma \ref{strongly connected is block 1} to prove that the image by any $\sigma \in G$ of the foliage of $\kblock{C_e}{k+1}$  is also a \scc\ and a \blim.$\qed$
\end{proof}

The previous lemmata form the proof by induction of:

\begin{proposition}
For any $k\geq 1$, the foliages of the vertices of each $\gamgr{k}{\aut{A}}$ form a system of imprimitivity.  
\end{proposition}

Note that for any $k \geq 1$ if it happens that $\leaf{\kblock{C_e}{k}} = Q$ then $\gamgr{k}{\aut{A}}$ is strongly connected and $\aut{A}$ is \compreach. Now we will prove that if this is not the case for any $k$, then some \blim\ that contains $e$ is invariant by $a$.  We will use the following set:
\begin{multline*}
D_k(\aut{A}) := \nset{p \in Q \mid p \in \dupl w \text{ for some } w \in \Sigmaast \\  \text{ such that } \vert \excl w \vert \leq k \text{ and } e \in \excl{w} \subseteq \leaf{\kblock{C_e}{k-1}}}.
\end{multline*}
The set of states duplicated by words of defect less than $k$ such that their excluded set is contained in $C_e^{[k]}$.

\subsection{Intermezzo}
Before we continue, it is necessary to present some definitions and results related with the theory of permutation groups that are used in the rest of this section.
Let $Q$ be a finite set and $G \subseteq S_{Q}$ be a subgroup of permutations of $Q$.
For any non-empty subset $P \subset Q$ consider the set of permutations:
$$\St{G}{P} := \nset{\sigma \in G \mid P \cdot \sigma = P}.$$
This is, the set of permutations of $G$ that preserve $P$ set-wise.
It can be easily proved that $\St{G}{P}$ is a subgroup of $G$. Let us call it the \emph{stabilizer} of the subset $P$.

Now consider an arbitrary but fixed system of imprimitivity of $G$ over $Q$, call it $\mathfrak{B}$.
The following fact is well known and we will omit the proof.

\begin{proposition}
Let $G$ be a group of permutations of a finite set $Q$.
Suppose that $G$ is transitive and $\mathfrak{B}$ is a system of imprimitivity.
If $B,C \in \mathfrak{B}$ are two different blocks of imprimitivity then $\St{G}{B}$ and $\St{G}{C}$ are conjugate subgroups of $G$.
\label{prop:stabilizers conjugates}
\end{proposition} 

For a subgroup $H$ of a group $G$, the \emph{core} of $H$, denoted by $\cor{H}$, is the intersection of all the conjugates of $H$ in $G$, i.e.,
$$\cor{H} := \bigcap_{\sigma \in G} \sigma^{-1}H\sigma.$$
Note that this subgroup is normal for $G$. 

Resuming with the transitive group $G$ of permutations of $Q$, Proposition \ref{prop:stabilizers conjugates} tell us that for every system of imprimitivity $\mathfrak{B}$ of $Q$ all the stabilizers of the blocks are conjugate.
Hence, the following definition makes sense.

\begin{definition}
Let $G$ be a subgroup of permutations of $Q$ and $\mathfrak{B}$ a system of imprimitivity of $Q$.
The core of $\mathfrak{B}$, denoted by $\cor{\mathfrak{B}}$, is the intersection of all the stabilizers of the blocks in $\mathfrak{B}$. 

\label{def:core system of impr}
\end{definition}

In some occasions it is more convenient to work with \blsim, hence to talk about the core of a block of imprimitivity.
If $\mathfrak{B}$ is a system of imprimitivity and  $B \in \mathfrak{B}$ is a block, we denote $\cor{B}:= \cor{\mathfrak{B}}$.
For our purposes we look for the core of certain \blsim\ to act in a transitive way on said blocks. 
We can ensure this if said core acts transitively on at least one of the blocks.

\begin{proposition}
Let $G$ be a group of permutations of the finite set $Q$.
Suppose that $G$ is transitive and $\mathfrak{B}$ is a system of imprimitivity.
If $B \in \mathfrak{B}$ is a block and $\cor{\mathfrak{B}}$ acts transitively on $B$, then this core is also transitive on all the blocks of $\mathfrak{B}$.
\label{prop:transitive on one}
\end{proposition}

\begin{proof}
Let $C\in \mathfrak{B}$ be a different block of $\mathfrak{B}$, besides let $p,q \in C$ be two different states.
We aim to prove that there is a permutation $\sigma \in \cor{\mathfrak{B}}$ such that $p \cdot \sigma = q$.
Being $G$ transitive, there is a permutation $\tau \in G$ such that $C \cdot \tau = B$.
Let $r,s \in B$ be such that $p\cdot \tau = r$ and $q \cdot \tau = s$.
By hypothesis, there is a permutation $\rho \in \cor{\mathfrak{B}}$ such that $r \cdot \rho = s$.
Thus
 $$p \cdot \tau\rho \tau^{-1} = q.$$
Since the core is normal in $G$ we can conclude that $\tau\rho \tau^{-1} \in \cor{\mathfrak{B}}$.
\end{proof}

\subsection{Non-reachability and invariance}

In this part we see that for some almost-group automata not being completely reachable implies there is at least one \blim\ invariant by the letter of defect 1.

Before the main proposition we present a technical lemma.
Since in the following lemma $k$ is arbitrary but fixed, $C_e^{[k]}$ will be referred just as $C_e$.

\begin{lemma}
Let $\aut{A} = \langle Q, \Sigma_0\cup \nset{a} \rangle$ be an almost-group automaton.
If in $\gamgr{k}{\aut{A}}$ there is an edge $C_e \rightarrow X$ and $\cor{\leaf{C_e}}$ is transitive for $C_e$, then for every state $q \in \leaf{X}$, there exists a word $v$ of defect $k$ such that $\excl{v} \subset \leaf{C_e}$ and $q \in \dupl{v}$.
\label{lemma:cover with dupl}
\end{lemma}

\begin{proof}
The edge $C_e \rightarrow X$ is produced by a word $w$ such that $\excl{w} \subset \leaf{C_e}$ and $\dupl{w} \cap \leaf{X} \neq \emptyset$.
Let $p \in \dupl{w} \cap \leaf{X}$ be arbitrary.
Since $\cor{C_e}$ is transitive, by Proposition \ref{prop:transitive on one} there is a permutation $\sigma \in \cor{C_e}$ such that $p\cdot \sigma =q$.
At the same time it is true that $C_e \cdot \sigma = C_e$, since the core is a subset of $\St{G}{C_e}$.
Therefore we have that $\excl{w\sigma} \subset \leaf{C_e}$ and $q \in \dupl{w\sigma}$.
As wanted.
\end{proof}

Using the Lemma \ref{edges gamma k} we also can conclude:
\begin{corollary}
If in $\gamgr{k}{\aut{A}}$ there is an edge $X \xrightarrow{w} Y$ and $\cor{\leaf{C_e}}$ is transitive for $C_e$.
Then for every state $q \in \leaf{Y}$, there exists a word $v$ of defect $k$ such that $\excl{v} \subseteq \leaf{X}$ and $q \in \dupl{v}$.
\label{lemma:cover with dupl general}  
\end{corollary}

With these two lemmas, we are ready for the main result of this part: 

\begin{theorem}
Let $\aut{A} = \langle Q, \Sigma_0\cup \nset{a} \rangle$ be an almost-group automaton.
Suppose $\gamgr{}{\aut{A}}$ is not strongly connected.
This means for some $k \geq 1$ it happens that $\gamgr{}{\aut{A}} = \gamgr{k}{\aut{A}}$; and  $\kblock{C_e}{k}  = \kblock{C_e}{j}$ for every $j \ge k$. 
Besides this, suppose that for every $\ell \leq k $ the cores $\cor{\kblock{C_e}{\ell}}$ are transitive on $\kblock{C_e}{\ell}$.
Then $\leaf{\kblock{C_e}{k}}$ is invariant by $a$.
\label{prop:a-invariance implies compreach}
\end{theorem}

\begin{proof}
We will use a, structurally, similar proof of the same fact for binary automata presented in \cite{casas2022binary}.
Suppose that $\kblock{C_e}{k} = \kblock{C_e}{k+1}$.
By induction on $0\leq \ell \leq k$ we will prove that
$$\leaf{\kblock{C_e}{\ell}} \cdot a \subseteq \leaf{\kblock{C_e}{k}}.$$

For $\ell = 0$ take $\kblock{C_e}{0} = \nset{e}$ hence the proposition is true in this case.

Our \textit{first induction hypothesis} is that $\leaf{\kblock{C_e}{\ell}} \cdot a \subseteq \leaf{\kblock{C_e}{k}}.$
By the construction of $\gamgr{\ell+1}{\aut{A}}$, for any $p \in \leaf{\kblock{C_e}{\ell +1}}$ there is a $X_m \in Q_\ell$ such that $p \in \leaf{X_m}$ and there is a path:
$$
\kblock{C_e}{\ell} \rightarrow X_1 \rightarrow X_2 \rightarrow \dots \rightarrow X_m
$$
in $\gamgr{\ell}{\aut{A}}$.

Now, by induction on the length of the path (the number $m>1$) the idea is to prove that $\leaf{X_m}\cdot a \subseteq \leaf{\kblock{C_e}{k}}$.

If $m = 1$, since there is an edge $\kblock{C_e}{\ell} \rightarrow X_1$ we use Lemma \ref{lemma:cover with dupl} to ensure that for $p \in \leaf{X_1}$ there is a word $w \in W_{\ell}(\aut{A})$ such that $ \excl{w} \subseteq \leaf{\kblock{C_e}{\ell}}$ and $p \in \dupl{w} \cap \leaf{X_1}$.
The defect of $wa$ is at most $\ell +1 \leq k+1$ and by the \textit{first induction hypothesis} $\excl{wa} \subseteq \leaf{\kblock{C_e}{k}}$ and $$p\cdot a \in \dupl{wa} \subseteq D_{k+1}(\aut{A})\subseteq \kblock{C_e}{k+1} = \kblock{C_e}{k},$$ proving what we wanted.

Now suppose that $m>1$ and $\leaf{X_{m-1}}\cdot a \subseteq \leaf{\kblock{C_e}{k}}$, i.e., the \textit{second induction hypothesis}.
By the Corollary \ref{lemma:cover with dupl general} for $p\in \leaf{X_m}$ there is a word $w\in W_\ell(\aut{A})$ such that $\excl{w} \subseteq \leaf{X_{m-1}}$ and $p \in \dupl{w}$.
If we apply the same argument as before, but this time using the \textit{second induction hypothesis} we can conclude that 
$$p\cdot a \in \dupl{wa} \subseteq D_{k+1}(\aut{A})\subseteq \kblock{C_e}{k+1} = \kblock{C_e}{k},$$
again, as intended.

Since $\kblock{C_e}{\ell +1}$ is a \scc\ of $\gamgr{\ell}{\aut{A}}$, thus its foliage is the union of the respective foliages of its vertices.
We have proved that $\leaf{\kblock{C_e}{\ell +1}}\cdot a \subseteq \leaf{\kblock{C_e}{k+1}} =\leaf{\kblock{C_e}{k}}$.
\end{proof}

The previous theorem proves that for certain almost group automata  not being completely reachable is equivalent to having a non-trivial imprimitivity block that is invariant under the letter of defect 1. 

\section{Conclusion}

We considered automata with an alphabet such that there is exactly one letter of defect 1 and the other letters are permutations over the state set. We found a necessary and sufficient condition to decide whether these automata are completely reachable. We saw that if the group generated by the permutations is primitive, then the automaton is completely reachable. On the other case, if the group is transitive and it has non trivial \blsim\ the condition depends on the behaviour of the letter of defect one over certain \blsim. The author believes that the additional condition stated in Theorem \ref{prop:a-invariance implies compreach}, the one stating the transitivity of the cores on the blocks of imprimitivity, can be omitted but more work on this direction must be done.
In any case these results generalize what was presented in \cite{casas2022binary} where the alphabet was binary since the automata presented in that article are almost group and the group generated by the permutation letter is the cyclic one, which is abelian and thus the additional condition is given.
Once decided whether or not an automaton is completely reachable, the next interesting question is to find a bound to the shortest word required to reach subsets of size $1 \leq k < n$. In \cite{ferens2022completely} it is stated that this bound is at most $2 n(n - k)$; but we believe that due to the strict structure of the considered automata the bound can be improved. Nevertheless this problem is open by the moment.

\bibliographystyle{plain}
\bibliography{Bibliography_almost_group}

\end{document}